\documentclass{article}
\usepackage[a4paper, portrait, margin=1.1811in]{geometry}
\usepackage[english]{babel}
\usepackage[utf8]{inputenc}
\usepackage{etoolbox}
\usepackage{caption}
\usepackage{booktabs}
\usepackage{xcolor} 
\usepackage[T1]{fontenc}
\usepackage{graphicx}
\usepackage[english]{babel}

\makeatletter
\newcommand*\bigcdot{\mathpalette\bigcdot@{1}}
\newcommand*\bigcdot@[2]{\mathbin{\vcenter{\hbox{\scalebox{#2}{$\m@th#1\bullet$}}}}}
\makeatother
\usepackage{hyperref}
\hypersetup{%
  colorlinks = true,
  linkcolor = magenta,
  citecolor = blue
}

\usepackage{setspace}
\usepackage{titlesec}
\usepackage{csquotes}
\usepackage{complexity}
\usepackage[bottom]{footmisc}
\newcommand{\eqdef}{\mathrel{\mathop:}=}
\usepackage{array}
\usepackage{makecell}
\usepackage{adjustbox}
\usepackage{multirow}
\usepackage{algorithm, caption} 
\usepackage{algpseudocode} 
\usepackage{float}
\usepackage{scrextend}
\usepackage{fancyhdr}
\usepackage{amsthm}
\usepackage{amsmath}
\usepackage{amsfonts}
\usepackage{mathtools}
\usepackage[utf8]{inputenc}
\usepackage[english]{babel}
\usepackage{amssymb}

\newtheorem{theorem}{Theorem}
\newtheorem{definition}{Definition}

\newtheorem{lemma}{Lemma}

\makeatletter
\renewcommand*\env@matrix[1][*\c@MaxMatrixCols c]{%
  \hskip -\arraycolsep
  \let\@ifnextchar\new@ifnextchar
  \array{#1}}
\makeatother

\fancypagestyle{plain}{
	\fancyhf{}

}

\makeatletter
\patchcmd{\@maketitle}{\LARGE \@title}{\fontsize{16}{19.2}\selectfont\@title}{}{}
\makeatother

\usepackage{authblk}

\newsavebox\affbox
\author{Ihyun Nam}

\titleformat{\author}{\normalfont\fontsize{10}{15}\bfseries}{\thesection}{1em}{}
\title{\textbf{\huge }\\
A Survey of Multivariate Polynomial Commitment Schemes \\}
\date{June 2023}
\begin{document}
\renewcommand{\thepage}{\arabic{page}}

\setlength{\parskip}{0.5em}

\maketitle

\begin{center}
    \textbf{Abstract}
\end{center}
A commitment scheme is a cryptographic tool that allows one to commit to a hidden value, with the option to open it later at requested places without revealing the secret itself. Commitment schemes have important applications in zero-knowledge proofs and secure multi-party computation, just to name a few. This survey introduces a few multivariate polynomial commitment schemes that are built from a variety of mathematical structures. We study how Orion is constructed using hash functions; Dory, Bulletproofs, and Vampire using the inner-product argument; Signatures of Correct Computation using polynomial factoring; DARK and Dew using groups of unknown order; and Orion+ using a CP-SNARK. For each protocol, we prove its completeness and state its security assumptions.\footnote{This survey was written as part of my independent undergraduate research in Cryptography at Stanford University (Winter, Spring 2023).}

\pagebreak
{
  \hypersetup{linkcolor=black}
  \tableofcontents
}
\pagebreak

\section{What is a commitment scheme?}
In an interactive protocol involving a prover and a verifier, a commitment scheme is a way in which the prover may prove its knowledge of a statement to the verifier. The simplest way to do so would be for the prover to send the statement directly to the verifier. However, what if the prover wants to keep their knowledge secret or if the statement is simply too large to send directly? Commitment schemes are useful in such cases as they allow the prover to \textit{commit} to the statement and send the resulting commitment instead, which is usually some mathematical structure. Later, the verifier may request the prover to reveal evaluations of the commitment at certain points and provide proof that the claimed evaluations are in fact the correct openings of the commitment. At the end of the interaction, the verifier either accepts or rejects the proof – without actually seeing the statement, ever.
\section{Multivariate polynomial commitments}
A multivariate polynomial commitment scheme is simply a commitment scheme in which the entity that the prover commits to is a polynomial involving more than one variables. A more formal definition follows.

\begin{definition} [Multivariate polynomial commitment scheme] \label{def:generic} A multivariate polynomial commitment scheme denoted by {\fontfamily{cmss}\selectfont PC} is defined by the following tuple of algorithms.

\begin{itemize}
    \item \textbf{{\fontfamily{cmss}\selectfont PC.Setup}}$(\lambda, l, D) \rightarrow$ ({\fontfamily{cmss}\selectfont ck, rk}). On input a security parameter $\lambda$, the number of variables $l \in \mathbb{N}$ in the polynomial, and a maximum degree bound $D \in \mathbb{N}$ of the polynomial, outputs a key pair ({\fontfamily{cmss}\selectfont ck, rk}).
    \item \textbf{{\fontfamily{cmss}\selectfont PC.Commit}}$(\phi,$ {\fontfamily{cmss}\selectfont ck)}$ \rightarrow \mathcal{C}$. On input an $l$-variate polynomial $\phi$ and {\fontfamily{cmss}\selectfont ck}, outputs a commitment $\mathcal{C}$ to $\phi$.
    \item \textbf{{\fontfamily{cmss}\selectfont PC.Open}}$(\phi,$ {\fontfamily{cmss}\selectfont ck}, $\Vec{x}) \rightarrow \pi$. On input {\fontfamily{cmss}\selectfont ck}, the polynomial $\phi$, and an evaluation point $\Vec{x}$ in $l$-dimension, outputs an evaluation proof $\pi$.
    \item {\fontfamily{cmss}\selectfont \textbf{PC.Check}(rk}$, \mathcal{C}, \Vec{x}, \Vec{v}, \pi) \rightarrow \{0,1\}$. On input {\fontfamily{cmss}\selectfont rk}, commitment $\mathcal{C}$, evaluation point $\Vec{x}$, $\Vec{v}$ the alleged evaluation of $\phi$ at $\Vec{x}$, and evaluation proof $\pi$, outputs 1 if the evaluation proof is correct and 0 otherwise.
\end{itemize}
\end{definition}

The details of these algorithms may differ in different implementations of a commitment scheme. For example, some schemes have separate online and offline setup stages while some schemes have none. However, in this survey, all commitment schemes are introduced in this generic model for completeness, except for when an algorithm is missing.

\subsection{Completeness and soundness}
We define the completeness and soundness properties of a (multivariate) polynomial commitment scheme as follows.

\begin{definition}[Completeness \cite{orion}] \label{def:completeness} A commitment scheme is said to be \textit{complete} if for every polynomial $\phi$ and evaluation point $\Vec{x}$, the following probability holds.

\begin{equation*}
Pr\begin{bmatrix}[c|c]
   & \mathcal{C} \leftarrow PC.Commit(\phi) \\
PC.Check(\pi, \Vec{x},y, \mathcal{C}) = 1 & \Vec{x}, y, \pi \leftarrow PC.Open(\phi, \Vec{x}, \mathcal{C})\\
    & y=\phi(\Vec{x})\\
\end{bmatrix} =1.
\end{equation*}

\end{definition}
Completeness states that if a prover is truthful about knowing a statement, then the verifier will \textit{always} accept the prover's proof. The same property is called correctness for argument systems. To define soundness, we first define a negligible function.

\begin{definition} [Negligible \cite{khurana}]
    A function $f:\mathbb{N} \rightarrow \mathbb{N}$ is said to be \textit{negligible} if for every $c \in \mathbb{N}$ there exists  $k_0 \in \mathbb{N}$ such that for all $k > k_0$, $|f(k)| < k^{-c}$ is true.
\end{definition}

\begin{definition}[Soundness] \label{def:soundness} A commitment scheme is said to be \textit{sound} if for every polynomial $\phi$, an evaluation point $\vec{x}$, proof $\pi$, and an efficient adversary $\mathcal{A}$, the following probability holds.    
\begin{equation*}
Pr\begin{bmatrix}[c|c]
   {\fontfamily{cmss}\selectfont PC.Check}(\pi, \Vec{x}, \phi(\Vec{x}), \mathcal{C}) =1 \land & \mathcal{C} \leftarrow {\fontfamily{cmss}\selectfont PC.Commit} (\phi) \\
   {\fontfamily{cmss}\selectfont PC.Check}(\pi', \Vec{x}, \phi(\Vec{x})', \mathcal{C}) =1 \land  & (\phi(\Vec{x}), \phi(\Vec{x})', \pi, \pi') \leftarrow \mathcal{A}(\phi, \Vec{x}, \mathcal{C})\\
   \phi(\Vec{x}) \neq \phi(\Vec{x})' & \\
\end{bmatrix} \leq negl(\lambda).
\end{equation*}

\end{definition}

Soundness states that no efficient adversary can produce two distinct valid openings to a commitment at the same point, except for with negligible probability. In other words, if the prover is falsely claiming to know a statement, then the verifier rejects the proof, except for when the prover actually somehow came up with the correct evaluation at the opening without knowing the statement, with negligible probability.

\subsection{Security properties}
A useful polynomial commitment scheme may also additionally achieve hiding and evaluation binding properties. The \textit{hiding} property states that the distributions of commitments for two distinct messages are computationally indistinguishable. In other words, commitments in a hiding commitment scheme reveal nothing about the committed messages.

\begin{definition}[Hiding \cite{bulletproof}]
    A commitment scheme is said to be \textit{hiding} if for any probabilistic polynomial time (PPT) adversary $\mathcal{A}$,
    \begin{equation*}
        \Bigg | Pr\begin{bmatrix}[c|c]
            & {\fontfamily{cmss}\selectfont pp} \leftarrow \text{{\fontfamily{cmss}\selectfont Setup}}(1^{\lambda})\\
            b=b' & (x_0, x_1) \leftarrow \mathcal{A}({\fontfamily{cmss}\selectfont pp}), b \xleftarrow[]{R} \{0,1\}, r \xleftarrow[]{R}\mathcal{R} \\
            & \mathcal{C} \leftarrow \text{{\fontfamily{cmss}\selectfont Commit}}(x_b;r), b' \leftarrow \mathcal{A}({\fontfamily{cmss}\selectfont pp}, \mathcal{C})
        \end{bmatrix} -\frac{1}{2} \Bigg|\leq negl(\lambda)
    \end{equation*}
    for security parameter $\lambda$.
\end{definition}

\textit{Evaluation binding} states that it must be infeasible for a prover to successfully argue that $f(\Vec{z})=y$ and $f(\Vec{z})=y'$ for a polynomial $f$, evaluation point $\Vec{z}$, and $y\neq y'$. 

\begin{definition}[Binding \cite{bulletproof}]
    A commitment scheme is said to be \textit{binding} if for any PPT adversary $\mathcal{A}$,
    \begin{equation*}
        Pr\begin{bmatrix}[c|c]
            \text{{\fontfamily{cmss}\selectfont Commit}}(x_0;r_0)=\text{{\fontfamily{cmss}\selectfont Commit}}(x_1;r_1) \land x_0 \neq x_1 & {\fontfamily{cmss}\selectfont pp} \leftarrow \text{{\fontfamily{cmss}\selectfont Setup}}(1^{\lambda})\\
            & x_0, x_1, r_0, r_1 \leftarrow \mathcal{A}(\text{{\fontfamily{cmss}\selectfont pp})} 
        \end{bmatrix} \leq negl(\lambda)
    \end{equation*}
    for security parameter $\lambda$.
\end{definition}

This survey states what properties each commitment scheme satisfies, without proving them.
\section{Common concepts and notations}
We first define some concepts and notations that are used uniformly throughout this survey.
\begin{itemize}
    \item \textbf{Vectors.} Vectors are indicated with an overhead arrow i.e. $\Vec{a}.$
    \item \textbf{zk-SNARK.} A zero-knowledge succinct non-interactive argument of knowledge is a cryptographic protocol that allows a prover to prove its knowledge to the verifier without revealing anything about the statement. Many of the commitment schemes discussed in this survey are in fact part of a (zk-)SNARK.
    \item \textbf{Trusted setup.} A setup procedure used to generate a shared secret between a prover and a verifier in a zero-knowledge proof system. A trusted setup destroys the secret after distribution.
    \item \textbf{Prover time, proof size, and verifier time.} These are the three common efficiency majors of a zero-knowledge protocol. \textit{Prover time} is the overhead of the prover to generate the proof, \textit{proof size} is the total communication cost between the prover and the verifier, and \textit{verifier time} is the time the verifier needs to verify evaluations of a commitment \cite{orion}.
    \item \textbf{Fiat-Shamir transformation.} This is a technique used to make an interactive protocol non-interactive. Informally, instead of having the verifier send a random challenge value to the prover in the {\fontfamily{cmss}\selectfont Open} algorithm, the prover samples a challenge value on its own using a random function.
\end{itemize}

\section{Orion: using hash functions}
Orion \cite{orion} is a zero-knowledge argument system that achieves $O(N)$ prover time and $O(\log^2N)$ verifier time based on hash functions. A hash function is a one-way function that maps inputs of arbitrary lengths to hash values of fixed length. While it is easy to compute the hash value of an input, it is computationally difficult to recover an input given its hash value. Orion is plausibly quantum-secure and can be made non-interactive via the Fiat-Shamir transformation.

Orion's novel commitment scheme is based on a linear-time argument presented by Bootle et al. \cite{bcg20} that commits to multilinear polynomials. Most importantly, the commitment scheme performs bit-wise encoding on some message $M$ using the encoding function $E_C$, the details of which are omitted, and a novel \hyperref[def:code switching]{code switching} algorithm $C_{CS}$. 

\begin{itemize}
    \item \textbf{{\fontfamily{cmss}\selectfont PC.Commit$(\phi) \rightarrow \mathcal{R}$}}. On input a multilinear polynomial $\phi$, output a Merkle root commitment $\mathcal{R}$. The coefficients of $\phi$ is $w$.
    \begin{itemize}
        \item Parse $w$ as a $k \times k$ matrix.
        \item Compute the tensor code encoding $C_1$ as a $k \times n$ matrix that encodes each row of $M$ using $E_C$. Compute $C_2$ as a $n \times n$ matrix that encodes each column of $C_1$ using $E_C$. 
        \item For $i\in[n],$ compute the Merkle tree root $Root_i=${\fontfamily{cmss}\selectfont Merkle.Commit}$(C_2[:,i]).$ 
        \item Compute a Merkle tree root $\mathcal{R}=${\fontfamily{cmss}\selectfont Merkle.Commit}$([Root_0, \cdots, Root_{n-1}])$. 
    \end{itemize}
    
    \item \textbf{{\fontfamily{cmss}\selectfont PC.Open}$(\phi, \Vec{x}, \mathcal{R}) \rightarrow c_{\gamma0}, y_{\gamma0}, c_1, y_1, C_1[:,idx], Root_{idx}$}. On input $\phi, \mathcal{R}$, and an evaluation point $\Vec{x}$, the prover additionally has the following values:
    \begin{itemize}
        \item $\gamma_0 \leftarrow $ a random vector from the verifier in $\mathbb{F}^k$ for some $k\in \mathbb{N}$,
        \item $r= r_0 \otimes r_1 \leftarrow$ evaluation point $\Vec{x}$ parsed as a tensor product,
        \item $y \leftarrow \phi(\Vec{x}$), and
        \item $\hat{I}, I \leftarrow t\in[n]$ indices randomly sampled by the verifier.
    \end{itemize}
    The prover computes the following: 
    \begin{itemize}
        \item $c_{\gamma0}=\sum_{i=1}^{k=1}\gamma0[i]C_1[i]$
        \item $y_{\gamma0}=\sum_{i=0}^{k-1}\gamma0[i]w[i]$
        \item $c_1=\sum_{i=0}^{k-1}r_0[i]C_1[i]$
        \item $y_1=\sum_{i=0}^{k-1}r_0[i]w[i]$
        \item For $idx \in \hat{I}$, compute and send to verifier $C_1[:,idx]$ and the Merkle tree proof of $Root_{idx}$ for $C_2[:,idx]$ under $\mathcal{R}$.
    \end{itemize}
    The prover additionally sends the following to the verifier:
    \begin{itemize}
        \item $R_{c_1}:=$ {\fontfamily{cmss}\selectfont Merkle.Commit$(c_1)$}
        \item $R_{\gamma_0}:=$ {\fontfamily{cmss}\selectfont Merkle.Commit$(c_{\gamma_0})$}
        \item $y := \langle y_1, r_1 \rangle$
        \item Output of $C_{CS}: C_2[I[j]], idx], c_2[I[j]], c_{\gamma0}[I[j]]$ for all $0<j<|\hat{I}|$ and all $idx \in \hat{I}$
        \item $\pi \leftarrow$ proof of $C_{CS}$
    \end{itemize}
    
    \item \textbf{{\fontfamily{cmss}\selectfont PC.Check}}$(\pi, \Vec{x}, y, \mathcal{R}) \rightarrow \{0,1\}.$ On input $\pi, \Vec{x}, y, $ and $\mathcal{R}$, output 1 if all checks below pass and 0 otherwise. For $idx \in \hat{I}$, check the following:
    
    \begin{itemize}
        \item $c_{\gamma0}[idx]\stackrel{?}{=}\langle \gamma0, C_1[:,idx]\rangle$
        \item $E_C(y_{\gamma0})\stackrel{?}{=}c_{\gamma0}$
        \item $c_1[idx]\stackrel{?}{=}\langle r_0, C_1[:,idx]\rangle $
        \item $E_C(y_1)\stackrel{?}{=}c_1$
        \item $y\stackrel{?}{=}\langle r_1,y_1\rangle$
        \item $E_C(C_1[:,idx]) \text{ consistent with } Root_{idx} \text{ and its Merkle tree proof.}$
    \end{itemize}
\end{itemize}

Code switching is an efficient instantiation of a proof composition technique introduced by Rothblum and Ron-Zewi \cite{rr20}. Code switching reduces the proof size of the commitment scheme introduced by Bootle et al. \cite{bcg20} and Brakedown \cite{brakedown} from $O(\sqrt{N})$ to $O(\log^2N)$. Orion calls a pre-defined zero-knowledge argument protocol on $C_{CS}$ and uses the outputs of the algorithm in {\fontfamily{cmss}\selectfont Check}. The code switching algorithm is defined as follows.
    \begin{algorithm}[H]
    \caption*{\textbf{The code switching algorithm $(C_{CS})$ \cite{orion}}\label{def:code switching}} \label{alg:codeswitching}
	\begin{algorithmic}[1]
            \State Encode $c_{\gamma0}\coloneqq E_C(y_{\gamma0}), c_1 \coloneqq E_C(y_1)$
            \For {idx $\in I$}
            \State Encode $C_2[:,idx] \coloneqq E_C(C_1[:,idx)$
            \For{idx $\in \hat{I}$}
            \State Check $c_{\gamma0} \stackrel{?}{=} \langle \gamma0, C_1[:,idx]\rangle $
            \State Check $c_1[idx]\stackrel{?}{=}\langle r_0,C_1[:,idx]\rangle $
    
            \State Check $\langle r_1,y_1\rangle \stackrel{?}{=}y$
    
            \For{$0\leq j\langle |I|$}
            \State Output $c_1[I[j]],c_{\gamma0}[I[j]]$
            \For{idx $\in I$}
            \State Output $C_2[I[j],idx]$
            \EndFor
            \EndFor
            \EndFor
            \EndFor
        \end{algorithmic} 
    \end{algorithm}

\subsection{Correctness and security of Orion} \label{def:orion security}
\begin{theorem} Orion is a polynomial commitment that achieves correctness as defined in \hyperref[def:completeness]{Definition 2.2}.
\end{theorem}

\begin{proof}
    For all $idx \in \hat{I}$, we have $E_C(C_1[:,idx]) = C_2[:,idx]$. {\fontfamily{cmss}\selectfont PC.Commit} computes \\{\fontfamily{cmss}\selectfont Merkle.Commit}$(E_C(C_1[:,idx]))$ as ${\fontfamily{cmss}\selectfont Merkle.Commit}(C_2[:,idx])$, which is consistent with $Root_{idx}$. The prover sends the Merkle tree proof of $Root_{idx}$ for $C_2[:,idx]$ under $\mathcal{R}$ to verifier, which is correct by construction.
\end{proof}

The security of Orion relies on the hiding property of the Merkle hash function. The {\fontfamily{cmss}\selectfont Merkle.Commit} algorithm is used in Orion to compute the Merkle tree roots $Root_i$ as a commitment to the tensor code encoding of the message $M$ and also to compute the final commitment $\mathcal{R}$. Because the hash algorithm is considered one-way, given an output $y$ to {\fontfamily{cmss}\selectfont Merkle.Commit} (denote this hash function as $H$) it is infeasible for an adversary to find $M$ such that $y=H(M)$. For simplicity the manipulation of $M$ using $E_C$ and double tensor-encoding is omitted in the statement $H(M)$.

\section{Dory: using generalized inner-product arguments}\label{def:dory}
Dory \cite{dory} is a transparent setup, public-coin interactive argument for proving completeness of an inner-pairing product. Importantly, Dory works on vector commitments involving elements from two distinct groups. The inner-product argument was initially suggested by Bootle et al. \cite{ipa} and was used to construct zero-knowledge proofs for arithmetic circuit satisfiability with efficient communication. The inner-product argument takes as inputs two independent generators $G_1, G_2 \in \mathbb{G}^n$, a scalar $c \in \mathbb{Z}_p$, and $P\in \mathbb{G}.$ A prover in this argument system wants to convince a verifier that the prover knows $\Vec{a}, \Vec{b} \in \mathbb{Z}_p^n$ such that
\[\{(G_1, G_2, P, c; \Vec{a}, \Vec{b}): P=G_1^{\Vec{a}} G_2^{\Vec{b}} \land c=\langle \Vec{a}, \Vec{b}\rangle\}.\]
Here, $P$ is the vector commitment to $\Vec{a}, \Vec{b}.$ A full tuple of algorithms for an inner-product commitment scheme is explained in \hyperref[def:appenA]{Appendix A.}

\begin{definition} [Generalized inner-product arguments]
    A generalized inner-product argument \\(GIPA) is a generalization of the inner-product argument. Under GIPA, a prover can additionally prove inner-products between bilinear pairings, not only field elements.
\end{definition}

Additionally, Dory composes Pedersen and AFGHO commitments to commit to matrices using two-tiered homomorphic commitments proposed by Groth \cite{groth_2011}.

\begin{definition} [Pedersen commitment \cite{bulletproof}\label{def:pedersen}] The Pedersen commitment is a linearly homomorphic scheme that defines a function $M_{pp} \times R_{pp} \rightarrow C_{pp}$ for message space $M_{pp}=\mathbb{Z}_p$, randomness space $R_{pp}=\mathbb{Z}_p$, and commitment space $C_{pp}=\mathbb{G}_p$. The {\fontfamily{cmss}\selectfont Setup} algorithm outputs $g,h \overset{{\scriptscriptstyle\R}}{\leftarrow} \mathbb{G}$ and {\fontfamily{cmss}\selectfont Commit} outputs $g^xh^r$.
\end{definition}

\begin{definition} [AFGHO commitment \cite{abe}]
     Given a commitment key $(v_0, v_1) \in \mathbb{G}_2$ the commitment to $(A_0, A_1) \in \mathbb{G}_1^2$ is given by $e(A_0, v_0)e(A_1, v_1)$.
\end{definition}

\subsection{Using {\fontfamily{cmss}\selectfont Dory-PC-RE} with {\fontfamily{cmss}\selectfont EVAL-VMV-RE}}
There are several sequential implementations of Dory, each with some more properties than the other. In this survey, we discuss {\fontfamily{cmss}\selectfont Dory-PC-RE}, which is an honest-verifier, statistical zero-knowledge, random evaluation extractable polynomial commitment scheme for $r$-variable multilinear polynomials, using a commitment to matrices along with the Pedersen commitment. {\fontfamily{cmss}\selectfont Dory-PC-RE} uses {\fontfamily{cmss}\selectfont EVAL-VMV-RE} to evaluate commitments. The resulting commitment scheme achieves linear prover time and logarithmic verifier time.

    Take $\mathcal{X}=\mathbb{F}^{n \times m}$ and $M_{ij}\in\mathcal{X}$. Then,
    \begin{itemize}
        \item {\fontfamily{cmss}\selectfont \textbf{Setup}}$(1^\lambda)\rightarrow {\text{\fontfamily{cmss}\selectfont pp}.}$ On input the security parameter $\lambda$, output $ \text{\fontfamily{cmss}\selectfont pp} := (\Gamma_1 \overset{{\scriptscriptstyle\R}}{\leftarrow} \mathbb{G}_1^m, H_1 \overset{{\scriptscriptstyle\R}}{\leftarrow} \mathbb{G}_1, \Gamma_2 \overset{{\scriptscriptstyle\R}}{\leftarrow} \mathbb{G}_2^n, H_2 \overset{{\scriptscriptstyle\R}}{\leftarrow} \mathbb{G}_2)$
        \item {\fontfamily{cmss}\selectfont \textbf{Commit}}$(\text{{\fontfamily{cmss}\selectfont pp}}; M_{ij})\rightarrow (\mathcal{C}, \mathcal{S}).$ On input {\fontfamily{cmss}\selectfont pp} and the $ij$-th entry in the message $M_{ij}$, output $(\mathcal{C}, \mathcal{S}) :=  (\mathcal{C}, (r_{rows}, r_{fin}, \Vec{V}))$ for 
        $\begin{cases}
                 r_{rows} \overset{{\scriptscriptstyle\R}}{\leftarrow} \mathbb{F}^n \\
                 r_{fin} \overset{{\scriptscriptstyle\R}}{\leftarrow} \mathbb{F} \\
                 H_T \leftarrow e(H_1, H_2) \\
                 V_i \leftarrow \text{{\fontfamily{cmss}\selectfont Commit}}_{\text{Pedersen}}((\Gamma_1, H_1); M_{ij}, r_{rows,i} \\
                 \mathcal{C} \leftarrow \text{{\fontfamily{cmss}\selectfont Commit}}_{\text{AFGHO}}((\Gamma_2, H_T); \Vec{V}, r_{fin}) \\ 
            \end{cases}$
        \item {\fontfamily{cmss}\selectfont \textbf{Check}}$(cc; \mathcal{C}, M, \mathcal{S}) \rightarrow \{0,1\}$. The verifier outputs \\$\begin{cases}
        1 & \text{if } \mathcal{C} = \sum_i \Gamma_{2i}\bigg(\sum_j M_{ij}\Gamma_{1j}+r_{rows,i}H_1\bigg)+r_{fin}\cdot e(H_1, H_2)\\
        0 & \text{if otherwise}\end{cases}$
    \end{itemize}

To see what happens in {\fontfamily{cmss}\selectfont Dory-PC-RE}, we need a bit more information on the Pedersen commitment. For $n=2^m$, the Pedersen commitment for $\mathbb{F}$ with public parameters $pp_\mathbb{F}$, and the matrix commitment for $\mathbb{F}^{n \times n}$ with public parameters $pp_{\mathbb{F}^{n \times n}}$ is such that:
\begin{align*}
    (\mathcal{C}_M, \mathcal{C}_y, \Vec{L}, \Vec{R}) &\in \mathcal{L}_{VMV, n, pp_{\mathbb{F}}^{n\times n}, pp_\mathbb{F}} \subset \mathbb{G}_T \times \mathbb{G}_1 \times \mathbb{F}^n \times \mathbb{F}^n \\
    &\Leftrightarrow \exists (M \in \mathbb{F}^{n\times n}, y\in \mathbb{F}, \mathcal{S}_M, \mathcal{S}_y):y = \Vec{L}^TM\Vec{R},\\
    & \hspace{1cm} \text{such that {\fontfamily{cmss}\selectfont Open(pp}}_{\mathbb{F}}^{n\times n}, \mathcal{C}_M, M, \mathcal{S}_M)=1, \text{{\fontfamily{cmss}\selectfont Open(pp}}_\mathbb{F}, \mathcal{C}_y, y, \mathcal{S}_y)=1. 
\end{align*}
 In {\fontfamily{cmss}\selectfont Dory-PC-RE}, fix Pedersen commitment parameters as {\fontfamily{cmss}\selectfont pp}$_\mathbb{F}=(\Gamma_{1, fin} \overset{{\scriptscriptstyle\R}}{\leftarrow} \mathbb{G}_1^{2m}\times \mathbb{G}_1)$, and matrix commitment parameters as {\fontfamily{cmss}\selectfont pp}$_{\mathbb{F}^{n\times n}}=(\Gamma_{1,0}\overset{{\scriptscriptstyle\R}}{\leftarrow} \Gamma_2^{2m}\times \mathbb{G}_2)$. The prover wants to prove that $(T, y_{com}, \Vec{L}, \Vec{R}) \in \mathcal{L}_{VMV,n,pp_{\mathbb{F}^{n\times n}}, pp_{\mathbb{F}}}$. Here, we have $y_{com}=y\Gamma_{1,fin}$ as the commitment to $y=\Vec{L}^TM\Vec{R}$ and $T$ as a hiding commitment to $\Vec{T}' \in \mathbb{G}_1^n$.
 
Once the prover computes the vector $\Vec{v}=\Vec{L}^TM$, we see that by construction $y=\Vec{L}^TM\Vec{R}=\langle \Vec{v},\Vec{R}\rangle .$ Also, it follows from the linear homomorphism of Pedersen commitments that \[v_{com} \coloneqq \langle \Vec{L}, C'\rangle =\text{{\fontfamily{cmss}\selectfont Commit}}_{\Gamma_{1,0}}(\Vec{v})\] is a hiding commitment to $\Vec{v}$. Therefore, in {\fontfamily{cmss}\selectfont EVAL-VMV-RE}, it suffices for the prover to prove knowledge of $\Vec{T}' \in \mathbb{G}_1^n$ and $\Vec{v}\in \mathbb{F}^n$ such that
 \begin{align*}
     T&=\langle \Vec{T}', \Gamma_2\rangle  \\
     \langle \Vec{L}, \Vec{T}'\rangle  &= \langle \Vec{v},\Gamma_1\rangle  \\
     y_{com} &= \langle \Vec{v}, \Vec{R}\rangle \Gamma_{1,fin}.
 \end{align*} We see that this satisfies correctness by construction. Some implementation of Dory also satisfies the Honest-Verifier Statistical Zero-Knowledge properties by achieving hiding.
 
\subsection{Security of Dory}
Dory relies on the Symmetric eXternal Diffie-Hellman assumption (SXHD) in the standard model for security.
\begin{definition}[SXDH]
    For $(\mathbb{F}_p, \mathbb{G}_1, \mathbb{G}_2, \mathbb{G}_T, e, G_1, G_2)$ defined as above, the decisional Diffie-Hellman (DDH) assumption holds for $(\mathbb{F}_p, \mathbb{G}_1, G_1)$ and $(\mathbb{F}_p, \mathbb{G}_2, G_2)$. In any group, DDH implies DLOG, double pairing assumption, and reverse double pairing assumption defined below.
\end{definition}

\begin{definition}[Double pairing assumption \cite{dory}]
    For $(\mathbb{F}_p, \mathbb{G}_1, \mathbb{G}_2, \mathbb{G}_T, e, G_1, G_2)$ defined as above, given $A_1, A_2 \xleftarrow[]{R}\mathbb{G}_1$ no non-uniform PPT adversary can compute non-trivial $B_1, B_2 \in \mathbb{G}_2$ such that $A_1B_1+A_2B_2=0$. Similarly, given $A_1, A_2 \xleftarrow[]{R}\mathbb{G}_2$ no adversary can compute non-trivial $B_1,B_2\in \mathbb{G}_1$ such that $B_1A_1+B_2A_2=0.$
\end{definition}

\begin{definition}[Reverse double pairing assumption\cite{dory}]
    For $(\mathbb{F}_p, \mathbb{G}_1, \mathbb{G}_2, \mathbb{G}_T, e, G_1, G_2)$ defined as above and $n=poly(\lambda)$, given $\Vec{A} \xleftarrow[]{R}\mathbb{G}_1^n$ no non-uniform PPT adversary can compute a non-trivial $\Vec{B}\in \mathbb{G}_2^n$ such that $\langle \Vec{A}, \Vec{B}\rangle=0$. Similarly, given $\Vec{A} \xleftarrow[]{R}\mathbb{G}_2^n$, no adversary can compute non-trivial $\Vec{B}\in \mathbb{G}_2^1$ such that $\langle \Vec{B}, \Vec{A}\rangle=0$.
\end{definition}

\section{Bulletproofs: using improved inner-product arguments}
Bulletproofs \cite{bulletproof} is a non-interactive zero-knowledge proof protocol that does not require a trusted setup. While Bulletproofs can be used to prove a variety of statements on a committed value, this survey focuses on how Bulletproofs can yield a range proof. That is, proving that a commitment $V$ contains a number $v$ that lies in some range, without revealing $v$. The proof size is logarithmic in the size of the witness. Bulletproofs have especially important applications to secret monetary transactions on distributed and trustless blockchains.

\subsection{Improved inner-product argument}
This is a modification of the \hyperref[def:dory]{inner-product argument} introduced by Bootle et al.\cite{bcg20}. Bulletproofs build a proof system for the relation:
\[\{(\Vec{g}, \Vec{h},u, P; \Vec{a}, \Vec{b}\in \mathbb{Z}_p^n): P=\Vec{g}^{\Vec{a}} \Vec{h}^{\Vec{b}}\cdot u^{\langle \Vec{a}, \Vec{b}\rangle}\}\hspace{0.5cm} (*).\] This is a special case of the original inner-product argument when $c=\langle \Vec{a}, \Vec{b}\rangle$ is given as part of the vector commitment $P$. Bulletproofs construct a proof protocol for (*) with $2\log_2(n)$ proof size, as opposed to a naive solution with $2n$ proof size when the prover directly sends $\Vec{a}, \Vec{b}\in \mathbb{Z}_p^n$ to the verifier. Importantly for $n>1$, {\fontfamily{cmss}\selectfont Open} function, Bulletproofs have the prover and the verifier engage in an following recursive inner-product argument for $P'$.
\\To see how, first define a hash function $H:\mathbb{Z}_p^{2n+1}\rightarrow \mathbb{G}$ that takes in as inputs $\Vec{a}, \Vec{a}', \Vec{b}, \Vec{b}' \in \mathbb{Z}_p^{n'}$ and $c\in \mathbb{Z}_p$, and outputs
\[H(\Vec{a}, \Vec{a}', \Vec{b}, \Vec{b}, c)=g_{[:n']}^a \cdot g_{[:n']}^{a'} \cdot h_{[:n']}^b \cdot h_{[:n']}^{b'}\cdot u^c \in \mathbb{G}.\]
Here, the notation $\Vec{a}_{[:n]}$ denotes slices of vectors $(a_1, \cdots, a_n) \in \mathbb{F}^n$. Bulletproofs implement an \textit{improved} inner-product argument that reduces the overall proof size by a factor of 3. The following tuple of algorithms defines the relation $(*)$, and therefore also the proof system for Bulletproofs.

\begin{itemize}
    \item {\fontfamily{cmss}\selectfont Step 1.} Prover computes the following and send $L, R$ to the verifier.
    \begin{itemize}
        \item $n' = \frac{n}{2}$
        \item $c_L=\langle \Vec{a}_{[:n']}, \Vec{b}_{[n':]}\rangle$
        \item $c_R=\langle \Vec{a}_{[n':]}, \Vec{b}_{[:n']}\rangle$
        \item $L=\Vec{g}_{[n':]}^{\Vec{a}_{[:n']}}\Vec{h}_{[:n']}^{\Vec{b}_{[n':]}}u^{c_L}$
        \item $R=\Vec{g}_{[:n']}^{\Vec{a}_{[n':]}}\Vec{h}_{[n':]}^{\Vec{b}_{[:n']}}u^{c_R}$
    \end{itemize}
    
    \item {\fontfamily{cmss}\selectfont Step 2.} Verifier randomly samples $x$ from $\mathbb{Z}_p^{\star}$ and sends it to the prover.
    \item {\fontfamily{cmss}\selectfont Step 3.} Now that they both have $x$, the prover and the verifier compute
    \begin{itemize}
        \item $\Vec{g'}=\Vec{g}_{[:n']}^{x^{-1}}\circ \Vec{g}_{[n':]}^x$
        \item $\Vec{h'}=\Vec{h}_{[:n']}^x\circ \Vec{h}_{[n':]}^{x^{-1}}$
        \item $P'=L^{x^2}PR^{x^{-2}}$
    \end{itemize}
    \item {\fontfamily{cmss}\selectfont Step 4.} Prover computes 
    \begin{itemize}
        \item $\Vec{a'}=\Vec{a}_{[:n']}\cdot x+ \Vec{a}_{[n':]}\cdot x^{-1}$
        \item $\Vec{b'}=\Vec{b}_{[:n']}\cdot x^{-1}+ \Vec{b}_{[n':]}\cdot x$
    \end{itemize}
    \item {\fontfamily{cmss}\selectfont Step 5.} Verifier outputs 1 if
    \[P'=\bigg(\Vec{g}_{[:n']}^{x^{-1}} \circ \Vec{g}_{[n':]}^x\bigg)^{\Vec{a'}} \cdot \bigg(\Vec{h}_{[:n']}^x \circ \Vec{h}_{[n':]}^{x^{-1}}\bigg)^{\Vec{b'}}\cdot u^{\langle \Vec{a'}, \Vec{b'}\rangle}\] and 0 otherwise.
    \item  {\fontfamily{cmss}\selectfont Step 5.} Repeat the protocol recursively on input $(\Vec{g'}, \Vec{h'}, u, P'; \Vec{a'}, \Vec{b'})$ so that the prover ends up sending $(L_1, R_1), \cdots, (L_{\log_2n}, R_{\log_2n}), (a,b).$
    
\end{itemize}

\subsection{Range proof using Bulletproofs}
We introduce an example of a range proof using Bulletproofs using insights from a \hyperref[def:pedersen]{Pedersen commitment}. Curious readers should read the original Bulletproofs paper to see how they can also be used to make zero-knowledge proofs for arithmetic circuits. The proof system for a range proof proves the following relation:
\[(g,h \in \mathbb{G}, V,n;v,\gamma\in\mathbb{Z}_p:V=h^{\gamma}g^v\land v\in[0,2^{-1}]) \hspace{1cm} (*)\]
for a group $\mathbb{G}$ and a Pedersen commitment $V\in\mathbb{G}$ to $v\in[0, 2^n-1]$ with randomness $\gamma$.
\begin{itemize}
    \item \textbf{{\fontfamily{cmss}\selectfont Setup}}$(a_L, a_R)$
    \item \textbf{{\fontfamily{cmss}\selectfont Commit$(\Vec{a}_L, \Vec{a}_R)\rightarrow A$.}} On input vectors $\Vec{a}_L\in \{0,1\}^n$ such that $\langle \Vec{a}_L, \Vec{2}^n\rangle=v$ and $\Vec{a}_R=\Vec{a}_L-\Vec{1}^n \in \mathbb{Z}_p^n$, output the commitment $A=h^{\alpha} \Vec{g}^{\Vec{a}_L}\Vec{h}^{\Vec{a}_R}\in \mathbb{G}$ for $\alpha$ randomly sampled from $\mathbb{Z}_p$. Also randomly sample 'blinding vectors' $\Vec{s}_L, \Vec{s}_R$ from $\mathbb{Z}_p^n$ and $\rho$ from $\mathbb{Z}_p$. Compute $S=h^{\rho}\Vec{g}^{\Vec{s}_L}\Vec{h}^{\Vec{s}_R} \in \mathbb{G}$. Send $A,S$ to the verifier.
    \item \textbf{{\fontfamily{cmss}\selectfont Open$(y,z)$}}. Verifier randomly samples $y,z\in\mathbb{Z}_p$ and sends them to the prover. At this point, the protocol defines the following polynomials:
    \begin{itemize}
        \item $l(X) = (\Vec{a}_L-z\cdot \Vec{1}^n)+\Vec{s}_L\cdot X \in \mathbb{Z}_p^n[X]$
        \item $r(X)=\Vec{y}^n\circ (\Vec{a}_R+z\cdot \Vec{1}^n+\Vec{s}_R\cdot X)+z^2\cdot \Vec{2}^n\in \mathbb{Z}_p^n[X]$
        \item $t(X)=\langle l(X),r(X)\rangle=t_0+t_1\cdot X+t_2\cdot X^2 \in \mathbb{Z}_p[X]$
    \end{itemize}
    The prover's job is reduced to convincing the verifier that $t_0=v\cdot z^2+\delta(y,z)$ where $\delta(y,z)=(z-z^2)\cdot\langle \Vec{1}^n,\Vec{y}^n\rangle-z^3\langle\Vec{1}^n,\Vec{2}^n \rangle \in \mathbb{Z}_p.$ To do so, it goes through another cycle of \textbf{{\fontfamily{cmss}\selectfont Commit$_2$}} and \textbf{{\fontfamily{cmss}\selectfont Open$_2$}}.
    \item \textbf{{\fontfamily{cmss}\selectfont Commit$_2(t_1, t_2)\rightarrow T_1, T_2$.}} On input $t_1, t_2$ and $\tau_1, \tau_2$ sampled randomly from $\mathbb{Z}_p$, output the commitments $T_i=g^{t_i}h^{\tau_i}$ for $i=\{1,2\}.$ Prover sends $T_1, T_2$ to the verifier.
    \item \textbf{{\fontfamily{cmss}\selectfont Open}$_2(x) \rightarrow \tau_x, \mu, \hat{t}, \Vec{l}, \Vec{r}.$} On input a random challenge $x \in \mathbb{Z}_p^{\star}$ from the verifier, the prover outputs the evaluation of $l$ and $r$ at $x$ along with a few others as follows.
    \begin{itemize}
        \item $\Vec{l}=l(x)=\Vec{a}_L-z \cdot \Vec{1}^n+\Vec{s}_L\cdot x$
        \item $\Vec{r}=r(x)=\Vec{y}^n \circ (\Vec{a}_R+z \cdot \Vec{1}^n +\Vec{s}_R\cdot x)+z^2\cdot \Vec{2}^n$
        \item $\hat{t}=\langle \Vec{1}, \Vec{r}\rangle $
        \item $\tau_x=\tau_2 \cdot x^2+\tau_1\cdot x+z^2\cdot \gamma$
        \item $\mu=\alpha+\rho\cdot x$
    \end{itemize}
    \item \textbf{{\fontfamily{cmss}\selectfont Check$(\Vec{l}, \Vec{r}, \tau_x, \mu, \hat{t}) \rightarrow \{0, 1\}$}}. On input $\Vec{l}, \Vec{r}, \tau_x, \mu, \hat{t}$, the verifier outputs 1 if all the checks below pass and 0 otherwise.
    \begin{itemize}
        \item $g^{\hat{t}}h^{\tau_x} \stackrel{?}{=} V^{z^2}\cdot g^{\delta(y,z)}\cdot T_1^x \cdot T_2^{x^2}$
        \item $A \cdot S^x \cdot \Vec{g}^{-z}\cdot (\Vec{h'})^{z\cdot \Vec{y}^n+z^2\cdot \Vec{2}^n} \stackrel{?}{=} h^{\mu}\cdot \Vec{g}^{\Vec{l}}\cdot (\Vec{h'})^{\Vec{r}}$
        \item $\hat{t} \stackrel{?}{=} \langle \Vec{l}, \Vec{r}\rangle$
    \end{itemize}
\end{itemize}
This can be made into a non-interactive proof through Fiat-Shamir transformation and additionally achieve knowledge soundness.

\subsection{Completeness and security of Bulletproofs}
\begin{theorem}
    Bulletproofs achieves completeness as defined in \hyperref[def:completeness]{Definition 2.2}.
\end{theorem}
\begin{proof}
    We want to show that $L^{(x^2)}\cdot P\cdot R^{(x^2)}\stackrel{?}{=}H(x^{-1}\Vec{a'}, x\Vec{a'}, x\Vec{b'}, x^{-1}\Vec{b'}, \langle \Vec{a'},\Vec{b'}\rangle)$ in {\fontfamily{cmss}\selectfont Check} holds for an honest prover. The LHS is equal to $L^{x^2}\cdot P\cdot R^{x^{-2}}$, which is equal to
    \[H(\Vec{a}_{[:n']}+x^{-2}\Vec{a}_{[n':]}, x^2\Vec{a}_{[:n']}+\Vec{a}_{[n':]}, x^2\Vec{b}_{[n':]}+\Vec{b}_{[:n']}, \Vec{b}_{[n':]}+x^{-2}\Vec{b}_{[:n']}, \langle\Vec{a'},\Vec{b'}\rangle).\]
    Using the properties of $H$ as defined, the RHS is equal to $H(x^{-1}\Vec{a'}, x\Vec{a'}, x\Vec{b'}, x^{-1}\Vec{b'}, \langle \Vec{a'},\Vec{b'}\rangle)$ as required.
\end{proof}

Bulletproofs also provides the arithmetic circuit protocol. This can be used with the improved inner-product argument to achieve perfect completeness and computational soundness under the discrete logarithm assumption. The security of Bulletproofs is reduced to the hardness of the discrete logarithm problem, which guarantees the perfectly hiding and computationally binding properties of Pedersen commitments.

\begin{definition}[Discrete logarithm problem\cite{dlog}]
    Given a group $G$, a generator $g$ of the group and an element $h$ of $G$, find the discrete logarithm to the base $g$ of $h$ in the group $G$.
\end{definition}

The discrete logarithm problem may or may not be hard. The Discrete logarithm relation that Bulletproofs relies on, however, is a special case when $n\leq 1$ and is considered hard. 
\begin{definition}[Discrete Logarithm Relation \cite{bulletproof}]
    For all PPT adversaries $\mathcal{A}$:
    \begin{equation*}
Pr\begin{bmatrix}[c|c]
   & \mathbb{G} = Setup(1^{\lambda})\\
   \exists a_i \neq 0 \land \Pi^n_{i=1} g_i ^{a_i}=1 & g_1, \cdots, g_n \xleftarrow[]{R} \mathbb{G} \\
   & a_1, \cdots, a_n \in \mathbb{Z}_p \leftarrow \mathcal{A}(G, g_1, \cdots, g_n) \\
\end{bmatrix} \leq \mu(\lambda).
\end{equation*}
\end{definition}

\section{Vampire: polynomial commitments combined with inner-product commitments}
Lipmaa et al. present Vampire \cite{vampire}: a pairing-based updatable and universal zk-SNARK. The key subroutine that Vampire uses for its commitment component is Count, a novel univariate sumcheck argument. Informally, a sumcheck argument is a sumcheck protocol that is used to succinctly prove knowledge of openings for certain commitments \cite{sumcheck}. Commitment schemes that allow this are called sumcheck-friendly. In a univariate sumcheck argument for multiplicative subgroups, given a finite field $\mathbb{F}$ and some multiplicative subgroup $\mathbb{H} \subset \mathbb{F}$, the prover convinces the verifier that the committed polynomial $f \in \mathbb{F}$ sums to the agreed upon value $v \in \mathbb{F}$ over $\mathbb{H}$.

Count improves on the online efficiency of the previously optimal Aurora \cite{aurora}. In Count, the prover only needs to outputs one ILV commitment \cite{ilv} (which is essentially a single group element) instead of two polynomial commitments (two group elements) in Aurora. More formally, the commitment scheme in Count is defined by the following tuple of algorithms. It has three separate  {\fontfamily{cmss}\selectfont Setup} phases called {\fontfamily{cmss}\selectfont PGen, KGen,} and  {\fontfamily{cmss}\selectfont Derive}.
\begin{itemize}
    \item {\fontfamily{cmss}\selectfont \textbf{PGen}$(1^{\lambda}) \rightarrow$ pp=$(p, \mathbb{G}_1, \mathbb{G}_2, \mathbb{G}_T, e, g_1, g_2)$.} On input the security parameter $\lambda$, {\fontfamily{cmss}\selectfont PGen} is a bilinear group generator that returns a prime $p$, three additive cyclic groups $\mathbb{G}_1, \mathbb{G}_2, \mathbb{G}_T$ of order $p$, a non-degenerate efficiently computable bilinear pairing $e: \mathbb{G}_1 \times \mathbb{G}_2 \rightarrow \mathbb{G}_T$, a generator $g_b$ for $\mathbb{G}$.
    
    \item {\fontfamily{cmss}\selectfont \textbf{KGen}(pp$, n_h, d, d_{gap})\rightarrow$ srs.} On input  {\fontfamily{cmss}\selectfont pp}, $n_h=|\mathbb{H}|, d$ the degree of the commitment polynomial $f\in \mathbb{F}[X]$, and $d_{gap}\geq n_h \cdot \lfloor d/n_h \rfloor,$ compute the following.
    \begin{itemize}
        \item $S_1(X) \leftarrow \{(X^i)^{d_{gap}+d}_{i=0:i\neq d_{gap}}\}$
        \item $S_2(X) \leftarrow \{1, X, (X^{d_{gap}-n_hi})^{\lfloor d/{n_h}\rfloor}_{i=0}\}$
        \item $\sigma\xleftarrow[]{R}\mathbb{F}^*$
    \end{itemize} 
    Finally output {\fontfamily{cmss}\selectfont srs} $\leftarrow (p, n_h, d, d_{gap}, [g(\sigma):g \in S_1(X)]_1, [g(\sigma):g \in S_2(X)]_2)$.
    
    \item{\fontfamily{cmss}\selectfont \textbf{Derive}(srs)} $\rightarrow$ {\fontfamily{cmss}\selectfont ek$_\mathcal{R}$, vk$_\mathcal{R}$.} On input {\fontfamily{cmss}\selectfont srs}, compute $S(X) \leftarrow \sigma_{i=0}^{\lfloor d/{n_h}\rfloor}X^{d_{gap}-n_hi}$. Set {\fontfamily{cmss}\selectfont ek$_{\mathcal{R}} \leftarrow $srs} and set {\fontfamily{cmss}\selectfont vk$_{\mathcal{R}} \leftarrow $(srs,$[S(\sigma)]_2, [\sigma^{d_{gap}}]_T$)}. Return ({\fontfamily{cmss}\selectfont ek$_{\mathcal{R}}$}, {\fontfamily{cmss}\selectfont vk$_{\mathcal{R}}$}).

\end{itemize}
Once the preprocessing is done, the prover has ({\fontfamily{cmss}\selectfont ek$_{\mathcal{R}}, ([f(\sigma)]_1, v_f), f)$}, and the verifier has \\({\fontfamily{cmss}\selectfont ek$_{\mathcal{R}}, ([f(\sigma)]_1, v_f)$)}. The prover sends to the verifier the commitment to \[[f_{ipc}(\sigma)]_1=f(X)S(X)-v_f/{n_h}\cdot X^{d_{gap}}.\] In the {\fontfamily{cmss}\selectfont Check} algorithm, the verifier accepts if and only if

\[[f(\sigma)]_1 \bigcdot [S(\sigma)]_2-[f_{ipc}(\sigma)]_1 \bigcdot [1]_2 = v_f/{n_h}\cdot [\sigma^{d_{gap}}]_T \hspace{1cm} (*).\]

Notice that $(*)$ is a direct adaptation of the generic evaluation algorithm in the inner-product commitment scheme  (\hyperref[def:appenA]{Appendix A}) where instead of committing to a vector $\mu$ the prover is committing to a multivariate polynomial $f$.

\subsection{Completeness of Count}
To see why (*) correctly checks the validity of evaluations, we first define a lemma. A full proof for this lemma is present in the original Vampire paper \cite{vampire}.
\begin{lemma} \label{def:vamcorrect}
    We have the following set up.
    \begin{itemize}
        \item $\mathbb{H}$: order $n_h > 1$ multiplicative subgroup of $\mathbb{F}^*$
        \item $d_{gap}$: fixed parameter by Vampire such that $d_{gap}\leq n_h \cdot \lfloor d/n_h \rfloor$ for $d>0$
        \item $f\in PolyPunc(d, d_{gap}, X):=\{f(X)=\sum_{i=0}^{d_{gap}+d}f_iX^i \in \mathbb{F}_{\leq d_{gap}+d}[X]:f_{d_{gap}}=0\}$
        \item $S(X):=\sum_{i=0}^{\lfloor d/n_h\rfloor}X^{d_{gap}-n_h i} \in \mathbb{F}_{\leq d_{gap}}[X]$
    \end{itemize}
    Then, $\sum_{\chi\in \mathbb{H}}f(\chi)=v_f$ and $deg(f)\leq d$ if and only if there exists $f_{ipc}\in PolyPunc_{\mathbb{F}}(d,d_{gap},X)$ such that
    \[f(X)S(X)-f_{idc}(X)=\frac{v_f}{n_h}\cdot X^{d_{gap}}. \hspace{1cm} (\star)\]
\end{lemma}
Then, (*) is a natural manipulation of $(\star)$ and thus completeness follows directly from \hyperref[def:vamcorrect]{Lemma 7.1}.

\subsection{Security assumptions of Vampire}
Count relies on the $(d_1, d_2)-PDL$ (Power Discrete Logarithm) assumption of {\fontfamily{cmss}\selectfont KGen}.
\begin{definition} [$(d_1, d_2)-PDL$ assumption \cite{vampire}]
    Let $d_1(\lambda), d_2(\lambda) \in$ {\fontfamily{cmss}\selectfont Poly}$(\lambda)$. {\fontfamily{cmss}\selectfont KGen} is  $(d_1, d_2)-PDL$ secure if for any non-uniform PPT $\mathcal{A}$, 
    \begin{align*}
        Adv^{PDL}_{d_1, d_2, KeyGen, \mathcal{A}}(\lambda) &:= Pr\bigg[\mathcal{A}\bigg(p,[(x^i)^{d_1}_{i=0}]_1, [(x^i)^{d_2}_{i=0}]_2\bigg) = x \bigg| p \leftarrow \text{{\fontfamily{cmss}\selectfont KGen}}\bigg(1^{\lambda}; x \xleftarrow[]{R} \mathbb{F}*\bigg)\bigg]\\
        &= negl(\lambda).
    \end{align*}
\end{definition}

Count not only reduces the communication cost of Aurora but also gets rid of a low-degree test, which the prover in Aurora uses to convince the verifier of the degree bound of some polynomial $R$ that is part of a linear combination for $f$. 

\section{Signatures of Correct Computation: using polynomial factoring}
Signatures of correct computation (SCC) \cite{scc} is a model to verify dynamic computations by outsourcing the function to evaluate to an untrusted server. An SCC scheme for multivariate polynomial evaluation is equivalent to a publicly verifiable multivariate polynomial commitment scheme. SCC relies on polynomial decomposition, which is an extension of the factoring theorem for univariate polynomials, as defined below.

\begin{definition} [Polynomial decomposition] Let $f(\Vec{x}) \in \mathbb{Z}_p[\Vec{x}]$ be an $n$-variate polynomial. For all $\Vec{a} \in \mathbb{Z}_p^n$, there exist polynomials $q_i(\Vec{x})\in \mathbb{Z}_p[\Vec{x}]$ such that the polynomial $f(\Vec{x})-f(\Vec{a})$ can be expressed as $f(\Vec{x})-f(\Vec{a})=\sum_{i=1}^n(x_i-a_i)q_i(\Vec{x}).$ There exists a polynomial-time algorithm that finds such $q_i(\Vec{x}).$
\end{definition}

$MC_m$ is a commitment scheme built by Lehmkuhl and Chiesa \cite{ryan} based on SCC that is additionally hiding. $MC_m$ commits to polynomials with a single degree bound and is a direct generalization of the polynomial commitment scheme in KZG \cite{kzg}. $MC_m$ is defined by the following tuple of algorithms when committing to an $l$-variate polynomial $p$ with maximum degree bound $D \in \mathbb{N}$. Furthermore, $\mathcal{W}_{l,D}$ is a set of all multisets of $\{1,\cdots,D\}$ with cardinality at most $D$. Notice that polynomial decomposition is used in {\fontfamily{cmss}\selectfont
Open} when the prover (\textit{source} in SCC) produces the evaluation point for the requested opening.

\begin{itemize}
    \item $MC_m$.{\fontfamily{cmss}\selectfont
Setup}$(\lambda, l, D) \rightarrow $ {\fontfamily{cmss}\selectfont
(ck, rk)}. Sample a bilinear group $\langle 
\text{group}\rangle  \leftarrow \text{{\fontfamily{cmss}\selectfont SampleGrp}}(1^\lambda)$. Parse $\langle \text{group}\rangle$ as a tuple $(\mathbb{G}_1, \mathbb{G}_2, \mathbb{G}_T, q, G_1, G_2, e).$ Sample random elements $\beta_a, \cdots, \beta_l$ from $\mathbb{F}_q$. Then, {\fontfamily{cmss}\selectfont ck} $\coloneqq (\langle \text{group} \rangle, \sum)$ and {\fontfamily{cmss}\selectfont rk}$ \coloneqq (D, \langle \text{group}\rangle , \beta_1 G_1, \cdots, \beta_l G_1)$ for \\$\sum=\{(\Pi_{i\in W}\beta_i)G\}_{W\in W_{l,D}} \in \mathbb{G}_1^{D^l}$.
    \item $MC_m$.{\fontfamily{cmss}\selectfont
Commit(ck}, $p)\rightarrow \mathcal{C}.$ Compute the commitment $\mathcal{C}=[c_i]_i^n$ for $c_i \coloneqq p_i(\beta)G.$
    \item $MC_m$.{\fontfamily{cmss}\selectfont
Open(ck, }$p, \Vec{x}, \xi)\rightarrow \pi.$ For an opening challenge $\xi$ and evaluation point $\Vec{x}\in\mathbb{F}_q^l$, compute the linear combination of polynomials $p(\Vec{X}) \coloneqq \sum_{i=1}^n \xi^ip_i(\Vec{X})$. Compute $\Vec{w} \coloneqq [w_j]_{j=1}^l$ such that
\[p(\Vec{X})-p(\Vec{x})=\sum_{j=1}^l(X_j-z_j)w_j(\Vec{X}).\] Compute evaluation proof as $\pi \coloneqq [w_j\beta G]_{j=1}^l.$
    \item $MC_m$.{\fontfamily{cmss}\selectfont
Check(rk,}$\mathcal{C},\Vec{x},v,\pi) \rightarrow \{0,1\}.$ For the opening challenge $\xi$, compute $C \coloneqq \sum_{i=1}^n \xi^i c_i$ and $v \coloneqq \sum_{i=1}^n \xi^i v_i$. Output 1 if the following equality holds and 0 otherwise: 
\[e(C-vG,H)=\Pi_{j=1}^l e(w_j(\beta)G, \beta_jH-z_jH).\]
\end{itemize}

\subsection{Correctness and security of SCC}

Correctness of $MC_m$ follows from polynomial decomposition and the bilinear properties of $e$.

\begin{theorem}
    The commitment scheme $MC_m$ achieves correctness as defined in \hyperref[def:completeness]{Definition 2.2}.
\end{theorem}

\begin{proof}
    This proof was directly adopted from the original paper with little modification. Let an efficient adversary $\mathcal{A}${\fontfamily{cmss}\selectfont (ck,rk)} choose $l$-variate polyomials $\Vec{p} \coloneqq [p_i]_{i=1}^n \in \mathbb{F}_q$, evaluation point $\Vec{x}$, and opening challenge $\xi$. For $\Vec{v}=\Vec{p}(\Vec{x})$, observe that
    \begin{align*}
        e(C-vG, H) &= e(\sum_{i=1}^n\xi^ic_i-(\sum_{i=1}^n\xi^iv_i)G,H)\\
        &= e((p(\beta)-p(\Vec{x}))G,H) \\
        &= e((\sum_{j=1}^l(\beta_j-x_j)w_i(\beta))G,H) \\
        &= \Pi_{j=1}^le((\beta_j-x_j)(w_j(\beta))G,H) \\
        &= \Pi_{j=1}^l e(w_j, \beta_jH-x_jH)
    \end{align*}
    as required by {\fontfamily{cmss}\selectfont Check} \cite{ryan}.
\end{proof}

The security of SCC relies on the following computational assumption.
\begin{definition}[Bilinear $l$-strong Diffie-Hellman assumption \cite{idk2}]
    Suppose $\lambda$ is the security parameter and let $(p,\mathbb{G}, \mathbb{G}_T,e,g)$ be a uniformly randomly generataed tuple of bilinear pairings parameters. Given the elements $g,g^t,\cdots, g^{t^l} \in \mathbb{G}$ for some $t$ chosen at random from $\mathbb{Z}_p^*$ and for $l=poly(\lambda)$, there is no PPT algorithm that can output the pair $(c,e(g,g)^{1/(t+c)})\in \mathbb{Z}_p^*/\/{-t/}\times \mathbb{G}_T$ except with negligible probability $negl(\lambda).$
\end{definition}

\section{DARK: using groups of unknown order}
Bunz et al. \cite{dark} built a polynomial commitment scheme based on Diophantine Argument of Knowledge (DARK) compilers. We conveniently refer to the resulting polynomial commitment scheme as DARK. DARK does not require a trusted setup and thus is used to build a new transparent proof system called Supersonic.

DARK is an instantiation of an abstract homomorphic commitment function $[\![*]\!]: \mathbb{Z}_p[X] \rightarrow \mathbb{S}$ for some prime $p$ and some set $\mathbb{S}.$ The following tuple of algorithms is a concrete implementation of $[\![*]\!]$. Importantly, the {\fontfamily{cmss}\selectfont Check} algorithm recursively reduces the degree of $f$ from $d$ to $\lfloor \frac{n}{2} \rfloor$. At every recursive step, the verifier holds a commitment $[\![f(X)]\!]$ to $f(X)$, and the prover is bound to the underlying polynomial. The key to DARK is that the order of $\mathbb{G}$ cannot be computed efficiently – hence the name group of unknown order.
\begin{itemize}
    \item {\fontfamily{cmss}\selectfont \textbf{DARK.Setup}} $(1^\lambda) \rightarrow ${\fontfamily{cmss}\selectfont pp}. On input the security parameter $\lambda$, sample $\mathbb{G} \overset{{\scriptscriptstyle\R}}{\leftarrow} GGen(\lambda)$ and $G \overset{{\scriptscriptstyle\R}}{\leftarrow} \mathbb{G},$ output {\fontfamily{cmss}\selectfont pp} as $(\lambda, \mathbb{G}, G, q)$ where $q$ is a large integer.
    \item {\fontfamily{cmss}\selectfont \textbf{DARK.Commit}(pp}, $f(X)) \rightarrow (\mathcal{C}; f(X), \hat{f}(X))$. On input {\fontfamily{cmss}\selectfont pp} and the commitment polynomial $f(X)$, compute and return  $C = \hat{f}(q)\cdot G$, $f(X)$, and $\hat{f(X)}$ is $f(X)$ "lifted" to an integer polynomial with bounded coefficients such that $\hat{f(X)} \mod p = f(X)$.
    \item {\fontfamily{cmss}\selectfont \textbf{DARK.Open}(pp}$, C, f(X), \hat{f}(X)).$ Check that $\hat{f}(X) \in \mathbb{Z}(q/2)[X], \hat{f}(q)\cdot G=C$ and $f(X)=\hat{f}(X) \mod p$.
    \item {\fontfamily{cmss}\selectfont \textbf{DARK.Check}.} The prover computes $C_L = \hat{f}_L(q,q^2,\cdots,q^{(2^{\mu'}-1)})\cdot G$ and \\$C_R=\hat{f}_R(q,q^2,\cdots,q^{(2^{\mu'}-1)})\cdot G$. Then, the verifier outsources the expensive checking of $C_L+q^{d'+1}\cdot C_R=C$ to the prover by using a proof of exponentiation protocol \cite{poe}.
\end{itemize} 

When used with a $\mu$-variate polynomial of degree $d$, the proof size and verification work for DARK is $\mu \log(d)$. That is, the prover and verifier times are logarithmic in size to the number of coefficients in the commitment polynomial.

\subsection{Completeness of DARK}
\begin{theorem}
    The commitment scheme DARK achieves completeness for $\mu$-linear polynomials in $\mathbb{Z}_p[X]$ as defined in \hyperref[def:completeness]{Definition 2.2}.
\end{theorem}
\begin{proof}
    This proof was directly adopted from the original paper with little modification. Proving the completeness of DARK can be reduced to proving that $b<\frac{q}{2}$ and $|f|\leq b.$ We show by induction that for each recursive call of the evaluation protocol, inputs satisfy the following constraints: $f(X_1, \cdots, X_\mu)$ is $\mu$-linear, $C=G^{f(\Vec{q})}, f(X)\in\mathbb{Z}(b)$, and $f(z_1, \cdots, z_\mu)=y \mod p.$ 
    
    In the first iteration of the evaluation protocol, the prover maps the coefficients of a polynomial $\Tilde{f}(X_1, \cdots, X_\mu) \in \mathbb{Z}_p$ to a $\mu$-linear polynomial $f(X_1, \cdots, C_\mu)$ with coefficients in $\mathbb{Z}(p-1).$ 
    In the final iteration of the evaluation protocol, the prover produces $f$, such that $|f|<b$ by construction.

    Then, in the $i$-th iteration of the protocol, the prover computes $(i-1)$-linear polynomials $f_L$ and $f_R$ such that $f_L(X_1, \cdots, X_{i-1})+X_if_R(X_1, \cdots, X_{i-1})=f(X_1, \cdots, X_i).$ It follows that $f(z_1, \cdots, z_i) \mod p=f_L(z_1, \cdots, z_{i-1})+z_if_R(z_1, \cdots, z_{i-1})\mod p=y_L+z_iy_R\mod p=y.$ Finally, the prover can construct an $(i-1)$-linear polynomial $f'=f_L+\alpha f_R \in \mathbb{Z}(2^\lambda b)$. This computes the correct values of $f, y$, and $b$ such that $|f|<b$ \cite{dark}.
    
    This proof can be applied to $\mu$-variate polynomials by setting $f$ as such and adjusting the coefficients space accordingly.
\end{proof}

\subsection{Security assumptions of DARK}
The security of DARK depends on the order, adaptive root, and r-fractional root assumptions defind below.

\begin{definition}[Order assumption \cite{idk}]\label{def:orderassump}
    For a group of unknown order $\mathbb{G} \leftarrow GGen(k)$, the Order assumption holds if for any adversary $\mathcal{A}$
    
    \begin{equation*}
    Pr\begin{bmatrix}[c|c]
       & GGen(k) \rightarrow \mathbb{G} \\
    w \neq 1 \land w^{\alpha}=1 & \mathcal{A}(\mathbb{G})\rightarrow (w,\alpha)\\
        & \text{where }|\alpha|<2^{poly(k)}\in\mathbb{Z}\text{ and }w\in \mathbb{G}\\
    \end{bmatrix} \leq negl(k).
    \end{equation*}
\end{definition}

The order assumption is implied by the adaptive root assumption defined below.

\begin{definition}[Adaptive root assumption \cite{wes}]\label{def:adaptive root}
    For a group of unknown order $\mathbb{G}\leftarrow \text{{\fontfamily{cmss}\selectfont GGen}}$, we say that the Adaptive root assumption holds for \text{{\fontfamily{cmss}\selectfont GGen}} if for any efficient adversaries $\mathcal{A}_0, \mathcal{A}_1$:
    \begin{equation*}
    Pr\begin{bmatrix}[c|c]
        & \text{{\fontfamily{cmss}\selectfont GGen}}(k) \rightarrow \mathbb{G} \\
        & \mathcal{A}_0(\mathbb{G})\rightarrow (w, state) \\
    z^l=w\neq 1& l \xleftarrow[]{R} Primes(k)\\
        & \mathcal{A}_1(l, state) \rightarrow z
    \end{bmatrix} \leq negl(k).
    \end{equation*}
\end{definition}

\begin{definition}[r-Fractional Root Assumption]\label{def:rfrac}
    For a group of unknown order, the r-Fractional Root Assumption holds for \text{{\fontfamily{cmss}\selectfont GGen}} if for any efficient adversary $\mathcal{A}$:
    \begin{equation*}
    Pr\begin{bmatrix}[c|c]
        & \text{{\fontfamily{cmss}\selectfont GGen}}(k) \rightarrow \mathbb{G} \\
        &\mathbb{G}\rightarrow g \\
    u^{\beta}=g^{\alpha}\land\frac{\beta}{gcd(\alpha, \beta)} \neq r^k,k\in \mathbb{N}   & \mathcal{A}(\mathbb{G},g)\rightarrow (\alpha, \beta, u)\\\
        & \text{where } |\alpha|<2^{poly(k)}, \\
        &|\beta|<2^{poly(k)}\in \mathbb{Z} \text{ and } u\in \mathbb{G}
    \end{bmatrix} \leq negl(k).
    \end{equation*}
\end{definition}

\section{Dew: using groups of unknown order}
Arun et al. introduces Dew-PC, a new polynomial commitment scheme with constant size evaluation proofs and logarithm verifier time based on a transparent inner-product commitment scheme. This commitment scheme is ultimately used to compile an information-theoretic proof system and yield Dew \cite{dew}: a transparent (meaning no trusted setup) constant-sized zk-SNARK. Dew-PC is defined by the following tuple of algorithms.

\begin{itemize}
    \item {\fontfamily{cmss}\selectfont \textbf{Setup}$(1^{\lambda}, D) \rightarrow$ pp=$(\lambda, \mathbb{G},g,PRG,p)$.} On input $k$ the security parameter and $D$ the maximum degree of the commitment polynomial, output:
    \begin{itemize}
        \item $\mathbb{G} \leftarrow \text{{\fontfamily{cmss}\selectfont GGen}}(\lambda)$
        \item $g \xleftarrow[]{\$}\mathbb{G}$
        \item A pseudorandom generator PRG: $\{0,1\}^\lambda \rightarrow \mathbb{Z}_p^{2\log(D+1)}$
        \item A large prime $p$ such that $len(p)=poly(\lambda)$
    \end{itemize}
    \item {\fontfamily{cmss}\selectfont \textbf{Commit}(pp,$D,f(X),l-1) \rightarrow (\mathcal{C},\textbf{f})$.} On input public parameter {\fontfamily{cmss}\selectfont pp}, a degree (l-1) polynomial f(X) in $\mathbb{Z}_p[X]$, output:
    \begin{itemize}
        \item The commitment $\mathcal{C} \eqdef g^{\sum{i=0}^{l-1}f_i\alpha^{2i}}$
        \item The coefficient vector \textbf{f} where $f_i$'s are the coefficients of the polynomial
    \end{itemize}
    \item {\fontfamily{cmss}\selectfont \textbf{Open}(pp,$D, f(X), l-1, C, \Tilde{\textbf{f}}) \rightarrow$ pp=$(\lambda, \mathbb{G},g,PRG,p)$.} Output 1 if all of the following checks pass, and 0 otherwise.
    \begin{itemize}
        \item $l-1 \leq D$
        \item $\sum_{i=0}^{l-1} \Tilde{f_i}\alpha^{2i}\in\mathbb{Z}$
        \item $\mathcal{C}=g^{\sum_{i=0}^{l-1}\Tilde{f_i}\alpha^{2i}}$
        \item $\Tilde{\textbf{f}}\in \{\frac{a}{b}|gcd(a,b)=gcd(b,p)=1, 0<b<p^{1+\log l}, |a/b|\leq (2l+1)\alpha\}^l$
        \item $\Tilde{f_i}=f_i \mod p$
    \end{itemize}
    \item {\fontfamily{cmss}\selectfont \textbf{Check}(pp,$D,C,l-1,x,v,f(X)) \rightarrow$ \textbf{logTEST}$(\mathcal{C}, l-1; f(X))$, \textbf{logIPP}$(\mathcal{C}, l-1, x, v;f(X))$}. Return 1 if both subprocesses \textbf{logTEST} and \textbf{logIPP} accept. In both protocols described below, PoKPE is an argument of knowledge for the relation $R_{PoKPE}$ in the generic group model. It is reproduced in \hyperref[def:appenA]{Appendix A}.
    \begin{itemize}
        \item \textbf{logTEST:} Replace the query vector \textbf{z} in \hyperref[def:fig1]{Figure 1} with 
        \[z_k \equiv z_{k_0, \cdots, k_{\log l-1}}:=\Pi^{\log l}_{j=1}x_{j,k_{j-1}}.\]
        Here, $x_1, \cdots, x_{\log l}$ are sampled from $\mathbb{Z}^2_p$ where $x_j=(x_{j,0}, x_{j,1})$, and $(k_0, \cdots, k_{\log l-1})$ is the base-2 representation of $0\leq k\leq l-1$.
        \item \textbf{logIPP:} Replace the query vector \textbf{q} in \hyperref[def:fig2]{Figure 2} with 
        \[q_k:=\Pi_{0\leq j\leq\log l-1}(x^{k_j2^j}\mod p)\]
        where $x\in \mathbb{Z}_p$.
    \end{itemize}
\end{itemize}

\begin{figure}[H]
\begin{center}
\fbox{\includegraphics[width=0.5\paperwidth]{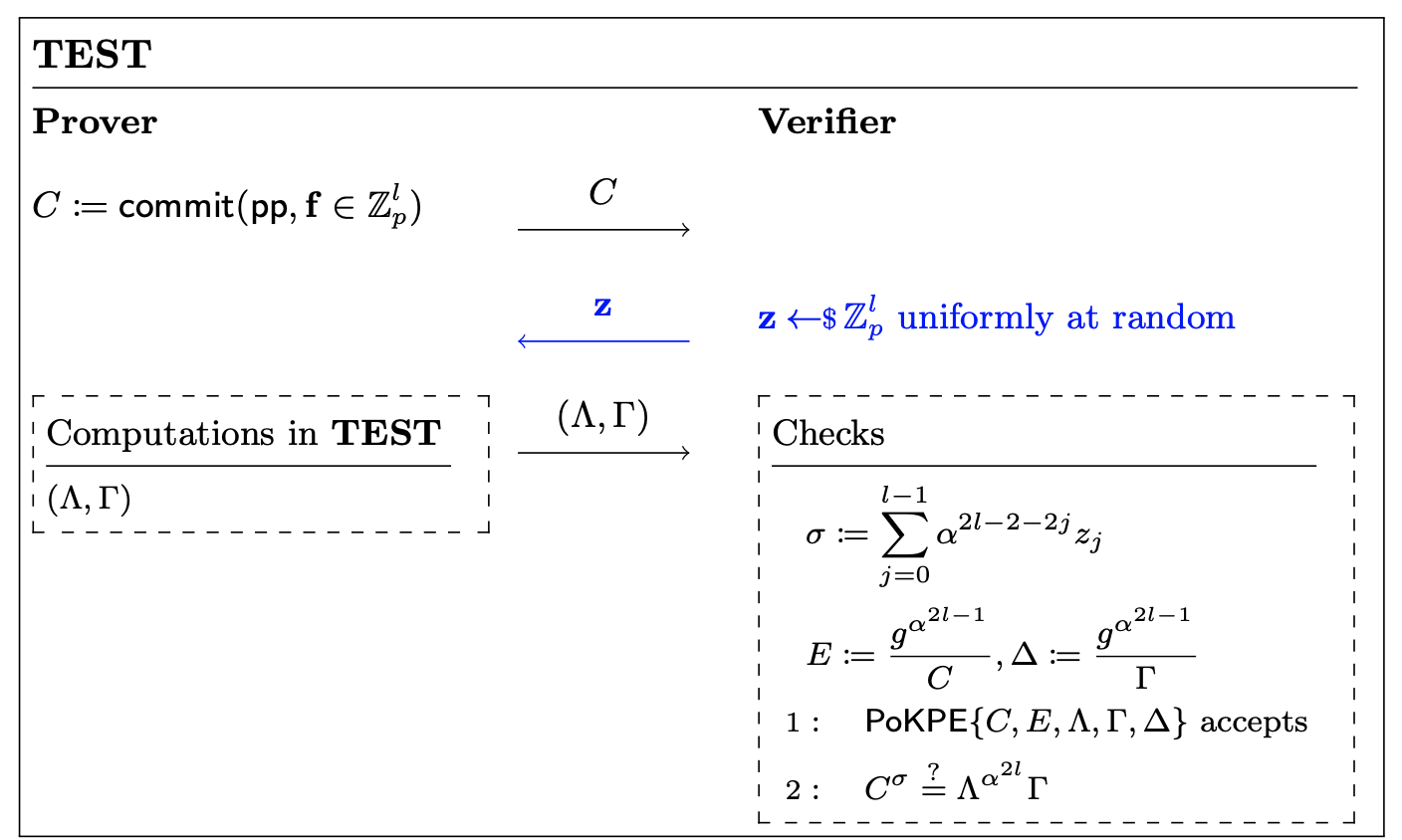}}
\caption{TEST protocol for ICP}
\label{def:fig1}
\end{center}
\end{figure}

\begin{figure}[H]
\begin{center}
\fbox{\includegraphics[width=0.5\paperwidth]{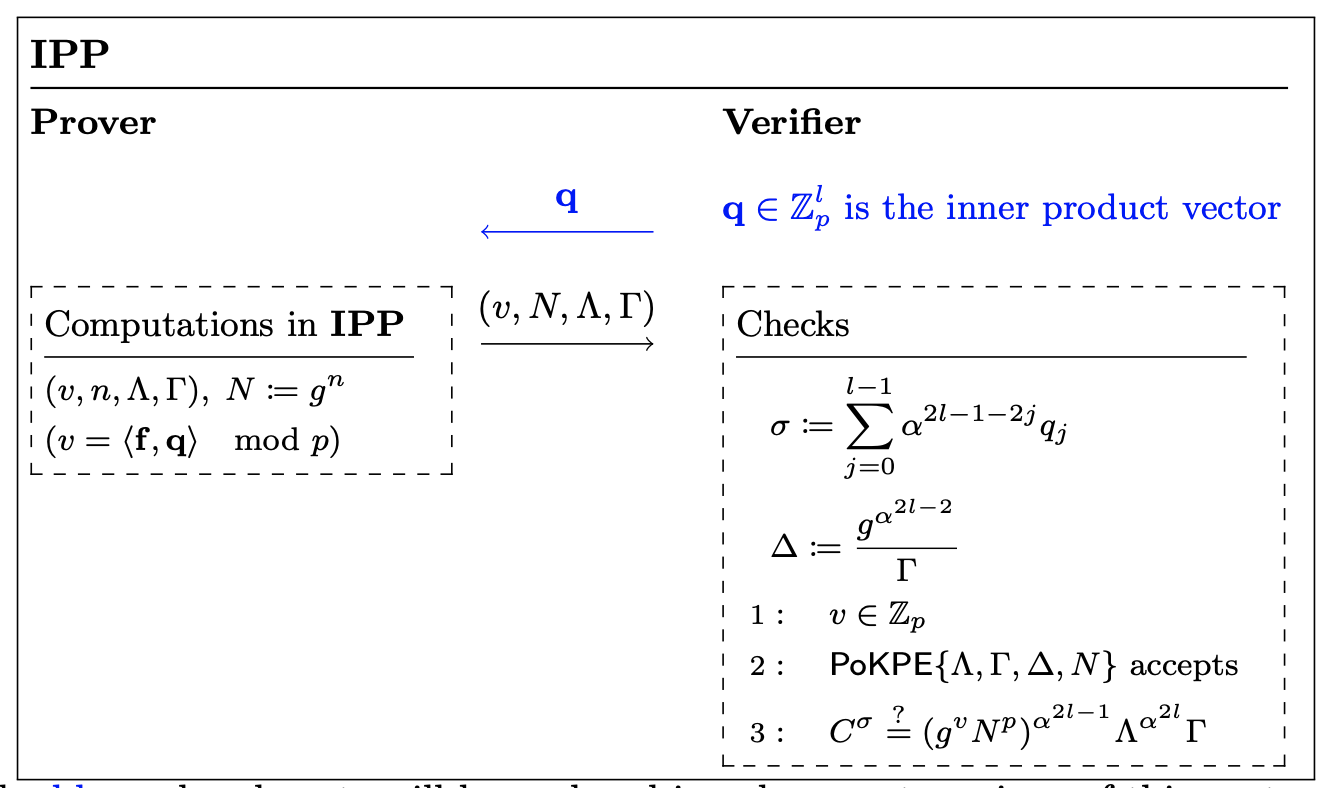}}
\caption{IPP protocol for ICP}
\label{def:fig2}
\end{center}
\end{figure}

The key component of achieveing logarithmic time verifier is to use Kronecker products of log-many length-2 vectors for \textbf{z} in \textbf{TEST} and \textbf{q} in \textbf{IPP} protocols. More details on these protocols may be found on the paper \cite{dew}.

Dew-PC can be made non-interactive using Fiat-Shamir transformation and the resulting scheme is succinct with respect to communication and verifier computation. 
\subsection{Completeness and security of Dew-PC}
\begin{theorem}
    The polynomial commitment scheme Dew-PC satisfies completeness as defined in \hyperref[def:completeness]{Definition 2.2}.
\end{theorem}
\begin{proof}
    This proof was directly adopted from the original paper with little modification. The completeness of Dew-PC relies on the definition of CoeffSplit and completeness of PoKPE. Coeffsplit is a protocol used in TEST and IPP of ICP and is reproduced in \hyperref[def:appenB]{Appendix B}. To show that the last check in \textbf{logTEST} holds, it suffices to show that $v=0$. Observe through direct manipulation that
    \begin{align*}
        &\sum_{j=0}^{l-1}f_j\alpha^{2j}\times \sum_{j=0}^{l-1}\alpha^{2l-2-2j}z_j \\
        &= \alpha^{2l}\bigg(\sum_{j'>j}\alpha^{2(j'-j)-2}f_{j'}z_j\bigg)+\bigg(\sum_{j'<j}\alpha^{2l-2-2(j-j')}f_{j'}z_j+\sum_{j'=j}\alpha^{2l-2}f_jz_j\bigg)\\
        &= \alpha^{2l}\lambda + \gamma.
    \end{align*}
    The $\lambda, \gamma$ in the end are indeed the return values of CoeffSplit and $v=0$ as required. For \textbf{logIPP} it suffices to show that $v=f(x)\mod p$. For the query vector \textbf{q}, observe that
    \begin{align*}
        &\sum_{j=0}^{l-1}f_j\alpha^{2j}\times \sum_{j=0}^{l-1}\alpha^{2l-1-2j}q_j\\
        &= \alpha^{2l-1}\bigg(\sum_{i=0}^{l-1}f_jq_j\mod p+p\bigg\lfloor\frac{\sum_{i=0}^{l-1}}{p}\bigg\rfloor\bigg) \\
        &+\alpha^{2l-2}(0)+\alpha^{2l}\bigg(\sum_{j'>j}\alpha^{2(j'-j)-1}f_{j'}q_j\bigg)+\bigg(\sum_{j'<j}\alpha^{2l-1-2(j-j')}f_{j'}q_j\bigg) \\
        &= \alpha^{2l-1}\bigg(\sum_{i=0}^{l-1}f_jq_j\mod p+p\bigg\lfloor\frac{\sum_{i=0}^{l-1}}{p}\bigg\rfloor\bigg)+\alpha^{2l-2}(0)+\alpha^{2l}\lambda+\gamma.
    \end{align*}
    The values in the end $(v+np,\lambda,\gamma)$ are indeed the output values of CoeffSplit and $v$ is satisfied.
\end{proof}

The security of Dew-PC check algorithm depends on the \hyperref[def:orderassump]{order assumption}, \hyperref[def:adaptive root]{adaptive root assumption}, and \hyperref[def:rfrac]{r-fractional root assumption} just like DARK, under the generic group model.

\section{Orion+: using a CP-SNARK}
Orion+ is a multilinear commitment scheme with linear prover complexity, $O(\mu)$ proof size, and $O(\mu)$ verifier time where for $\mu$-variate polynomials. Orion+ assists Hyperplonk \cite{hyperplonk}, a multilinear-IOP, by reducing the size of the opening proof by x1000 as compared to Orion \cite{orion}. Hyperplonk is an improvement on Plonk \cite{plonk}, a universal fully-succinct zk-SNARK of 2022 that achieved up to 20 times lower prover time than Sonic \cite{sonic}, by relying on boolean hypercubes. Importantly, Orion uses a CP-SNARK called OSNARK and is therefore a CP-SNARK itself.

\begin{definition}[CP-SNARK]
    A Commit-and-Prove SNARK (CP-SNARK) is a SNARK that aims to prove that, for vectors $\Vec{f}\in \mathbb{F^n}$ and $\Vec{g}\in \mathbb{F}^n$ with $m\leq n$ and their respective commitments $C_f$ and $C_g$, the values of $\Vec{f}(I_f)$ are equal to the values of $\Vec{g}(I_g).$ Here, $I_f \subseteq [n]$ and $I_g \subseteq[m].$ 
\end{definition}
Orion+ is defined by the following tuple of algorithms.
\begin{itemize}
    \item {\fontfamily{cmss}\selectfont \textbf{Setup}$(1^\lambda, \mu^*)\rightarrow gp(gp_o, gp_{pc},vp_o,pp_o).$} On input security parameter $\lambda$, upper bound $\mu^*$ on the number of variables.
    \begin{itemize}
        \item $m^* \leftarrow$ Set such that running time of {\fontfamily{cmss}\selectfont OSNARK} and {\fontfamily{cmss}\selectfont PC} is $O_\lambda(2^{\mu^*})$ for circuit size $m^*$.
        \item {\fontfamily{cmss}\selectfont $gp_o \leftarrow $ }{\fontfamily{cmss}\selectfont OSNARK.Setup}$(1^{\lambda},m^*)$
        \item {\fontfamily{cmss}\selectfont $gp_{pc} \leftarrow $} {\fontfamily{cmss}\selectfont PC.Setup}$(1^{\lambda},m^*)$
        \item {\fontfamily{cmss}\selectfont $vp_o, pp_o \leftarrow $ }run the indexing phase of {\fontfamily{cmss}\selectfont OSNARK}
    \end{itemize}
    \item {\fontfamily{cmss}\selectfont \textbf{Commit}(gp,$f) \rightarrow \mathcal{C}_f$.} On input polynomial $f \in \mathcal{F}_\mu^{(\leq 1)}$ with $\Vec{w}=(f_{\langle 0\rangle}, \cdots, f_{\langle n-1 \rangle})$ define the following.
    \begin{itemize}
        \item $m=n/k$
        \item $\Vec{w} \in \mathbb{F}^{k\times m}$
        \item Matrix $W \in \mathbb{F}^{k\times M}$ such that $W[i,:]=E(\Vec{w}[i,:])\forall i\in [k]$ for $E: \mathbb{F}^m \rightarrow \mathbb{F}^M$ the linear encoding.
        \item The hash commitment $h_j \leftarrow HCom(W[:,j])$ for each $j\in [M]$, where $W[:,j]$ is the $j$-th column of $W$.
        \item $p_h$ is the polynomial that interpolates vector $(h_j)_{j\in[M]}$
    \end{itemize} Output the commitment $\mathcal{C}_f \leftarrow$ {\fontfamily{cmss}\selectfont PC.Commit}$(gp_{pc}, p_h)$.
    \item {\fontfamily{cmss}\selectfont \textbf{Open}(gp,$\mathcal{C}_f,f)$}. Run {\fontfamily{cmss}\selectfont Commit} and see if the output is consistent with the input $\mathcal{C}_f$.
    \item {\fontfamily{cmss}\selectfont \textbf{Check}(gp;$\mathcal{C}_f, \Vec{z}, y;f) \rightarrow \{0,1\}.$} On input public parameters gp, point \textbf{z}$\in\mathbb{F}^{\mu}$ and commitment $\mathcal{C}_f$, transform \textbf{z} to vectors $\Vec{t_0} \in \mathbb{F}^k$ and $\Vec{t_1} \in \mathbb{F}^m$ such that f(\textbf{z})=$\langle \Vec{w}, \Vec{t_0}\otimes\Vec{t_1}\rangle$. The prover and the verifier run the following protocol and return 1 on completeness, taken directly from the original paper.
        
    \begin{figure}[H]
        \centering
        \fbox{\includegraphics[width=0.6\paperwidth]{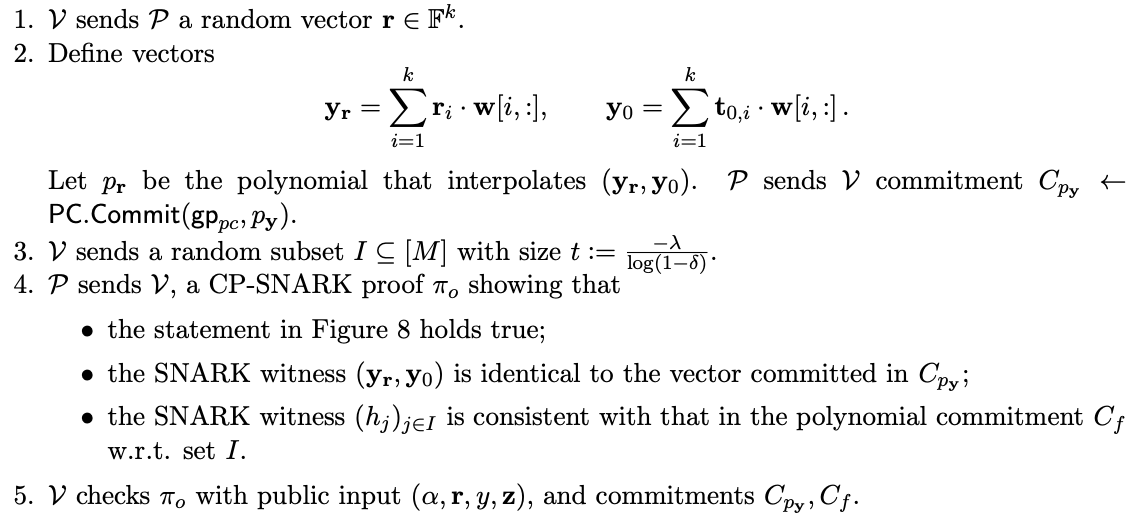}}
        \caption{The interactive sub-protocol for {\fontfamily{cmss}\selectfont Check} \cite{hyperplonk}}
    \end{figure}
\end{itemize}
Observe in the {\fontfamily{cmss}\selectfont Check} function that Orion+ uses OSNARK to commit $(\Vec{y_r}, \Vec{y_0})$ to a multilinear commitment $\mathcal{C}$ and build a CP-SNARK proof for their relation, without having to directly check their consistency. This helps Orion+ reduce prover time.

\subsection{Completeness and security of Orion+}
Orion+ is correct from inspection. Orion+ is made possibly quantum-resistant by not relying on the discrete log problem or bilinear pairing. The security of Orion+ instead relies on that of \hyperref[def:orion security]{Orion}, defined in the random oracle model, and OSNARK. Additionally, Orion+ is binding and sound.
 \pagebreak
\section{Appendix A}\label{def:appenA}
This is a generic tuple of algorithms for an inner-product commitment scheme, as presented in Vampire \cite{vampire}.
\begin{itemize}
    \item {\fontfamily{cmss}\selectfont Setup$(1^{\lambda}) \rightarrow$ PGen$(1^{\lambda}).$}
    \item {\fontfamily{cmss}\selectfont KeyGen(p, $N) \rightarrow$ (ck, tk).} On input the public parameter {\fontfamily{cmss}\selectfont p} and a vector dimension $N$, output
    \begin{itemize}
        \item {\fontfamily{cmss}\selectfont ck} $\leftarrow ([(\sigma^i)^{2N}_{i=0:i \neq N}]_1, [(\sigma^i)^N_{i=0}]_2)$
        \item {\fontfamily{cmss}\selectfont tk} $\leftarrow \sigma \xleftarrow[]{R} \mathbb{Z}_p^*$
    \end{itemize}
    \item {\fontfamily{cmss}\selectfont Commit(ck, $\mu) \rightarrow [\mu(\sigma)_1.]$} On input the commitment key {\fontfamily{cmss}\selectfont ck} and the vector $\mu$ that the prover wishes to commit to, return the commitment $[\mu(\sigma)_1]=\sum^N_{j=1}\mu_j[\sigma^j]_1$ where $\mu(X) \leftarrow \sum_{j=1}^N\mu_j X^j$.
    \item {\fontfamily{cmss}\selectfont Open (ck, $[\mu(\sigma)]_1, \mu, \Vec{v}) \rightarrow (v, [\mu_{ipc}(\sigma)]_1$.} On input the commitment key {\fontfamily{cmss}\selectfont ck}, the original vector $\mu$, another vector $\Vec{v}$, compute
    \begin{itemize}
        \item $v \leftarrow \mu^T\Vec{v}$
        \item $v^*(X) \leftarrow \sum^N_{j=1}\Vec{v_j}X^{N+1-j}$
        \item $\mu_{ipc}(X) \leftarrow \mu(X)\Vec{v}^*(X)-vX^{N+1}$
        \item $[\mu_{ipc}(\sigma)]_1 \leftarrow \sum^{2N}_{i=1, i\neq N+1}\mu_{ipc}[\sigma^i]_1$
    \end{itemize}
    Finally output $(v, [\mu_{ipc}(\sigma)]_1)$.
    \item {\fontfamily{cmss}\selectfont Check(ck, $[\mu(\sigma)]_1, \Vec{v}, v, [\mu_{ipc}(\sigma)]_1)\rightarrow \{0,1\}$.} Output 1 if 
    \[[\mu_{ipc}(\sigma)]_1 \bigcdot [1]_2 = [\mu(\sigma)]_1 \bigcdot \sum^N_{j=1}\Vec{v_j}[\sigma^{N+1-j}]_2-v[\sigma^N]_1\bigcdot [\sigma]_2\]
    and 0 otherwise.
\end{itemize}

\section{Appendix B}\label{def:appenB}
The PoKPE protocol as used in Dew. 
\begin{figure}[H]
    \centering
    \fbox{\includegraphics[width=0.5\paperwidth]{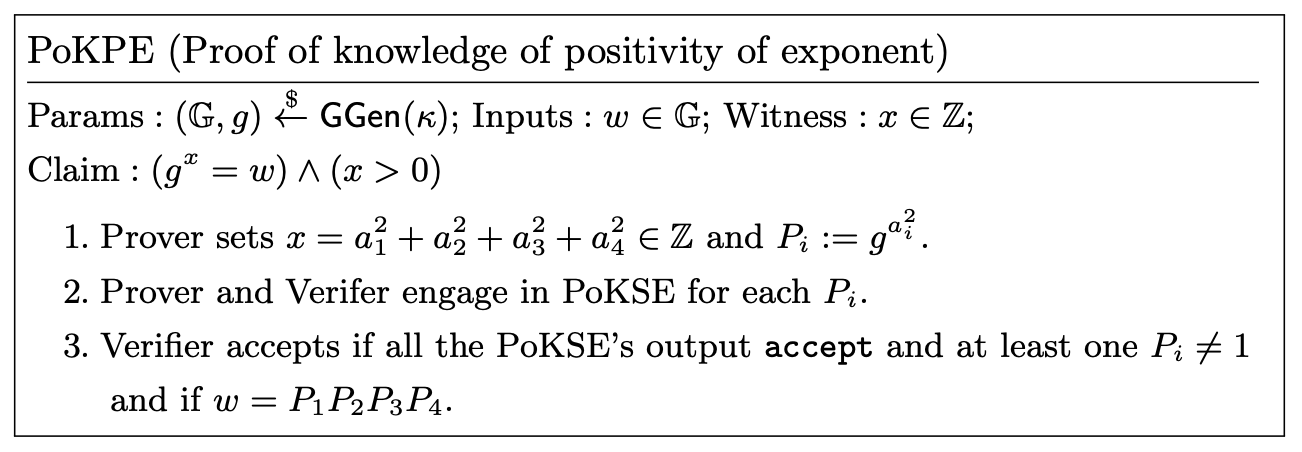}}
    \caption{Protocol PoKPE \cite{dew}}
\end{figure}

The CoeffSplit algorithm is defined as follows in Dew.
\begin{figure}[H]
    \centering
    \fbox{\includegraphics[width=0.5\paperwidth]{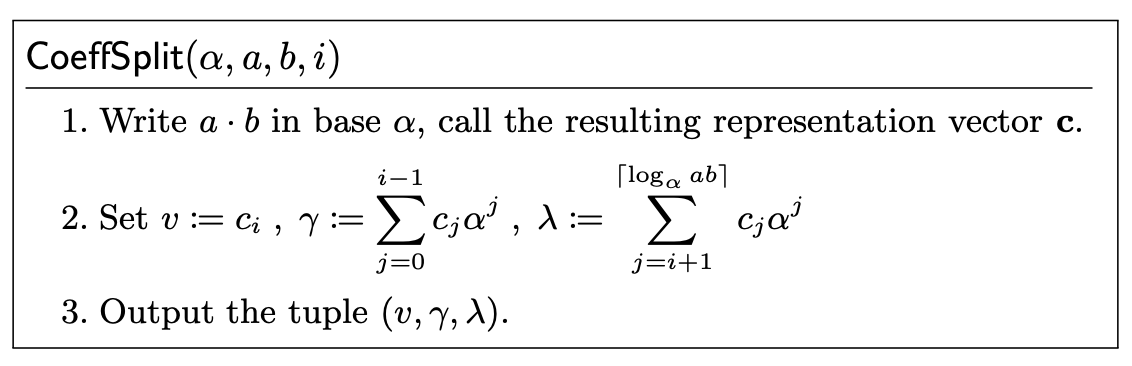}}
    \caption{CoeffSplit algorithm \cite{dew}}
\end{figure}

\bibliographystyle{IEEEtran}
\pagebreak

\bibliography{ref}

\end{document}